\documentclass[11pt,letterpaper,draftcls,oneside,peerreview,onecolumn]{IEEEtran}
\usepackage{amsfonts}
\usepackage{amssymb}
\usepackage{amsmath}
\usepackage[dvips]{graphicx}
\usepackage{cite}
\usepackage{balance}
\usepackage{caption}

\newtheorem{theorem}{Theorem}

\newtheorem{definition}{Definition}

\newtheorem{proposition}{Proposition}

\begin{document}

\title{Error Performance of Multidimensional Lattice Constellations - Part II: Evaluation over Fading Channels}

\author{Koralia N. Pappi, \IEEEmembership{Student Member, IEEE,}  Nestor D. Chatzidiamantis, \IEEEmembership{Student Member, IEEE,} and George~K.~Karagiannidis, \IEEEmembership{Senior Member, IEEE}

\thanks{The authors are with the Electrical and Computer Engineering
        Department, Aristotle University of Thessaloniki, GR-54124 Thessaloniki,
        Greece (e-mails: \{kpappi, nestoras,
        geokarag\}@auth.gr).}}
\maketitle

\begin{abstract}

This is the second part of a two-part series of papers, where the
error performance of multidimensional lattice constellations with
signal space diversity (SSD) is investigated. In Part I, following
a novel combinatorial geometrical approach which is based on
parallelotope geometry, we have presented an exact analytical
expression and two closed-form bounds for the symbol error
probability (SEP) in Additive White Gaussian Noise (AWGN). In the
present Part II, we extend the analysis and present a novel
analytical expression for the Frame Error Probability (FEP) of
multidimensional lattice constellations over Nakagami-$m$ fading
channels. As the FEP of infinite lattice constellations is lower
bounded by the Sphere Lower Bound (SLB), we propose the Sphere
Upper Bound (SUB) for block fading channels. Furthermore, two
novel bounds for the FEP of multidimensional lattice
constellations over block fading channels, named Multiple Sphere
Lower Bound (MSLB) and Multiple Sphere Upper Bound (MSUB), are
presented. The expressions for the SLB and SUB are given in closed
form, while the corresponding ones for MSLB and MSUB are given in
closed form for unitary block length. Numerical and simulation
results illustrate the tightness of the proposed bounds and
demonstrate that they can be efficiently used to set the
performance limits on the FEP of lattice constellations of
arbitrary structure, dimension and rank.
\end{abstract}


\newpage

\section{Introduction}\label{Intro}
The performance evaluation of multidimensional signal sets has
attracted significant attention due to the signal space diversity
(SSD) that these constellations present \cite{Boutros1} and the
fact that they can be efficiently used to combat the signal
degradation caused by fading. The design of such constellations
has been extensively studied in
\cite{Boutros2,Giraud,Viterbo2,Viterbo3}, but since the analytical
computation of the Voronoi cells of multidimensional
constellations is difficult\cite{Viterbo}, their error performance
has been evaluated only through approximations and bounds
\cite{Belfiore,Taricco,Bayer}, while for special cases, some exact
but complicated analytical expressions were derived \cite{Jihoon}.

In Part I \cite{Pappi} of this two-part series of papers, based on
parallelotope geometry we introduced a novel combinatorial
geometrical approach for the evaluation of the error performance
of multidimensional lattice constellations in Additive White
Gaussian Noise (AWGN). Especially, we proposed an exact analytical
expression for the Symbol Error Probability (SEP) of these signal
sets and two novel closed-form bounds, namely the Multiple Sphere
Lower Bound (MSLB) and the Multiple Sphere Upper Bound (MSUB).
With the introduction of the MSLB in part I, the concept of the
Sphere Lower Bound (SLB) was extended to the case of finite signal
sets. The SLB dates back to Shannon's work \cite{Shannon} and
although it has been thoroughly investigated in the literature
\cite{Viterbo,Tarokh,Vialle,Fabregas1}, it is not generally a
reliable lower bound for the important practical cases of finite
lattice constellations. Moreover, a similar upper bound, the
Sphere Upper Bound (SUB) has been investigated in \cite{Viterbo}
for AWGN channels.

\subsection{Contribution}\label{contribution}
In the present Part II, we study the error performance of
multidimensional infinite and finite lattice constellations in
Nakagami-$m$ block fading channels. Specifically, for infinite
lattice constellations:
\begin{itemize}
\item We propose a novel expression for the SUB which is suitable for the
analysis in fading channels while it upper bounds the Frame Error
Probability (FEP).
\item We present novel closed-form expressions for the well known SLB
and the proposed SUB in Nakagami-\emph{m} block fading channels.
\end{itemize}
For multidimensional lattice constellations, based on the proposed
expressions for the exact SEP, the MSLB and MSUB in AWGN given in
Part I \cite{Pappi}:
\begin{itemize}
\item We present a novel analytical expression for the Frame Error
Probability (FEP) of finite lattice constellations in the presence
of Nakagami-\emph{m} block fading.
\item Starting from this expression we
propose alternative formulae for the MSLB and the MSUB which are
suitable for the performance analysis in fading channels and bound
the FEP.
\item We present closed-form expressions
for the MSLB and MSUB in Nakagami-\emph{m} block fading channels
for the case of unitary block length.
\end{itemize}

\subsection{Structure}\label{structure}
The remainder of the paper is organized as follows. In Section
\ref{System_Model_Section}, the channel model and the
characteristics of faded lattices are presented. Section
\ref{perf_eval} investigates the exact FEP of infinite and finite
lattice constellations, while the SLB, MSLB, SUB and MSUB for
block fading are presented and their closed-form expressions are
proposed. Section \ref{results} illustrates the accuracy and
tightness of the proposed bounds via extensive numerical and
simulation results, whereas conclusions are discussed in Section
\ref{conclusions}.

\subsection{Notations}\label{Notations}

Here, we revisit some symbols and terms defined in Part I
\cite{Pappi} and also used in Part II:
\begin{itemize}
\item $\Lambda$ denotes an infinite lattice constellation and $\Lambda'$ a finite lattice constellation, carved from a
lattice $\Lambda$.
\item $N$ denotes the dimension of a lattice or lattice
constellation.
\item $\mathbf{M}$ denotes a generator matrix of a lattice
$\Lambda$, where $\mathbf{M}=\left[
\mathbf{v}_{1}~\mathbf{v}_{2}...\mathbf{v}_{N}\right]$,
$\mathbf{M}\in \mathbb{R} ^{N\times N}$ and
$|\mathrm{det}(\mathbf{M})|=1$. Vectors $\mathbf{v}_i$,
$i=1,2,\ldots,N$, are the basis vectors of the lattice.
\item $K$ is the number of symbols along the direction of
each basis vector.
\item $\mathcal{S}_N$ denotes the set of the basis vectors of the
$N$-dimensional lattice and $\mathcal{S}_{k,p}$ is a subset of $k$
out of $N$ basis vectors, with $p$ an index enumerating the
different possible subsets for each $k$. For specific $k$, the
index is $p=1,\ldots,\binom{N}{k}$.
\item $\mathcal{V}_{\mathcal{S}_{k,p}}$ denotes the Voronoi cell
of the sublattice, defined by the vector subset
$\mathcal{S}_{k,p}$.
\item $\mathrm{vol}_k(\cdot)$ is the volume of a $k$-dimensional
geometrical region. Note that
$\mathrm{vol}_N\left(\mathcal{V}_{\mathcal{S}_{N}}\right)=|\mathrm{det}(\mathbf{M})|=1$.
\item $d_{min}$ is the minimum distance between two points in an
infinite or in a finite lattice constellation.
\end{itemize}

\section{System Model}\label{System_Model_Section}

\subsection{Channel Model}\label{channel_model}

Let us consider a flat fading channel whose discrete time received
vector is given by
\begin{equation}
\mathbf{y}_{l}=\mathbf{Hx}_{l}+\mathbf{z}_{l},~~l=1,...,L,
\label{channel}
\end{equation}%
where $\mathbf{y}_{l}\in\mathbb{R}^{N}$ is the $N$-dimensional
real received signal vector, $\mathbf{x}_{l}\in \mathbb{R}^{N}$ is the $N$-dimensional real transmitted signal vector, $\mathbf{H}=$%
diag$\left( \mathbf{h}\right) \in \mathbb{R} ^{N\times N}$ is the
flat fading diagonal matrix with $\mathbf{h}=\left(
h_{1},...,h_{N}\right) \in \mathbb{R} ^{N}$, and
$\mathbf{z}_{l}\in \mathbb{R} ^{N}$ is the Additive White Gaussian
Noise (AWGN) vector whose samples are zero-mean Gaussian
independent random variables with variance $\sigma ^{2}$.
Furthermore, $L$ denotes the number of $N$-dimensional modulation
symbols in one frame.

The fading matrix $\mathbf{H}$ is assumed to be constant during
one frame and changes independently from frame to frame, i.e.,
block fading channel with $N$ blocks is considered. Thus, for a
given channel realization, the channel transition probabilities
are given by
\begin{equation}
p\left( \mathbf{y}\left\vert \mathbf{x},\mathbf{H}\right. \right)
=\left(
2\pi \sigma ^{2}\right) ^{-\frac{N}{2}}\exp \left( -\frac{1}{2\sigma ^{2}}%
\left\Vert \mathbf{y}-\mathbf{Hx}\right\Vert ^{2}\right) .
\label{probability}
\end{equation}%
Moreover, it is assumed that the real fading coefficients, $h_{i}$ for $%
i=1,...,N,$ follow Nakagami-$m$ distribution \cite{Alouini}, with
probability density function (pdf)\ given by
\begin{equation}
p_{h_{i}}\left( x\right) =\frac{2m^{m}x^{2m-1}}{\Gamma \left(
m\right) }\exp \left( -mx^{2}\right) ,  \label{Nakagami}
\end{equation}%
while the coefficients, $\gamma _{i}=h_{i}^{2}$, that correspond
to the fading power gains and will be used in the following
analysis, are Gamma distributed with pdf
\begin{equation}
p_{\gamma _{i}}\left( x\right) =\frac{m^{m}x^{m-1}}{\Gamma \left( m\right) }%
\exp \left( -mx\right)  \label{power_pdf}
\end{equation}%
and cumulative density function (cdf)
\begin{equation}
P_{\gamma _{i}}\left( x\right) =1-\frac{\Gamma \left( m,mx\right)
}{\Gamma \left( m\right) }.  \label{power_cdf}
\end{equation}%
In the above equations, $m\geq0.5$ and $\Gamma \left( \cdot
\right)$, $\Gamma \left( \cdot ,\cdot \right) $ denote the Gamma
\cite[Eq. (8.310)]{Gradshteyn} and the upper incomplete Gamma
\cite[Eq.
(8.310)]{Gradshteyn} functions, respectively. Finally, the signal-to-noise ratio (SNR) is defined as $\rho =\frac{1}{%
\sigma ^{2}}$.

\subsection{Faded Lattices}\label{faded_lattices}

As described in Part I, the transmitted signal vectors
$\mathbf{x}$ belong to an $N$-dimensional infinite or finite
lattice constellation, defined respectively as \cite[Eq.
(1)]{Pappi} \cite[Eq. (5)]{Pappi}
\begin{equation}\label{infinite_lattice_rule}
\Lambda=\mathbf{M}\mathbf{z},\,\,\,\,\,\,\,\,\mathbf{z}\in\mathbb{Z}^N,
\end{equation}
and
\begin{equation}\label{finite_lattice_rule}
\begin{array}{ccc}
\Lambda'=\mathbf{M}\mathbf{u},&\mathbf{u}=[u_1\,\,u_2\,\,\ldots\,\,u_N]^T,&u_i\in\{0,1,\ldots,K-1\}.
\end{array}
\end{equation}
Similarly, the faded infinite or finite lattice constellation is
defined as the lattice seen by the receiver which is given by
\begin{equation}\label{faded_inf}
\Lambda_f=\mathbf{H}\mathbf{M}\mathbf{z},\,\,\,\,\,\,\,\,\mathbf{z}\in\mathbb{Z}^N,
\end{equation}
and
\begin{equation}\label{faded_fin}
\begin{array}{ccc}
\Lambda'_f=\mathbf{H}\mathbf{M}\mathbf{u},&\mathbf{u}=[u_1\,\,u_2\,\,\ldots\,\,u_N]^T,&u_i\in\{0,1,\ldots,K-1\}.
\end{array}
\end{equation}
Accordingly, for the lattices in (\ref{faded_inf}) and
(\ref{faded_fin}) we define the faded generator matrix as
\begin{equation}\label{fad_gen}
\mathbf{M_f=HM.}
\end{equation}

All Voronoi cells on both infinite and finite lattice
constellations are distorted by fading. As a result, they are
dependent on the fading matrix $\mathbf{H}$. We denote a faded
Voronoi cell as $
\mathcal{V}_{\mathcal{S}_{k,p}}\left(\mathbf{H}\right)$.

\section{Performance Evaluation over Fading
Channels}\label{perf_eval}
\subsection{Frame Error Probability of Infinite and Finite
Lattice Constellations}\label{exact_SEP}

For the reader's convenience, we first present the exact
expressions for the Symbol Error Probability (SEP) of infinite and
finite lattice constellations in AWGN channels, as provided in
Part I \cite{Pappi}. For an infinite lattice constellation
$\Lambda$, the SEP is given by \cite[Eq. (12)]{Pappi}
\begin{equation}\label{inf_sep}
P_\infty(\rho)=1-\int_{\mathcal{V}_{\mathcal{S}_{N}}}
p(\mathbf{z})\mathrm{d}\mathbf{z}=1-J_N,
\end{equation}
whereas for a $K$-PAM lattice constellation it is given by
\cite[Eq. (17)]{Pappi}

\begin{equation}\label{KPAM_sep}
P_{K-PAM}(\rho)=1-\frac{\sum\limits_{k=0}^N(K-1)^k\sum\limits_{p=1}^{\binom{N}{k}}J_{k,p}}{K^N},
\end{equation}
with \cite[Eq. (16)]{Pappi}
\begin{equation}\label{Jkm}
J_{k,m}=\int_{\mathcal{V}_{\mathcal{S}_{k,p}}}
p(\mathbf{z})\mathrm{d}\mathbf{z},\,\,\,\,\,\,0<k<n,
\end{equation}
and $J_0=1$. For $k=0$ or $k=N$, it is $p=1$ and $p$ is omitted.
Furthermore, the frame error probability (FEP) can be written in
terms of the SEP, $P\left(\rho\right)$, as
\begin{equation}\label{pf_ps}
P_f\left(\rho\right)=1-\left(P_c\left(\rho\right)\right)^L=1-\left(1-P\left(\rho\right)\right)^L,
\end{equation}
where $P_c\left(\rho\right)$ is the probability of correct
reception.

The expressions in (\ref{inf_sep}) and (\ref{KPAM_sep}) are also
valid for a specific channel realization, i.e. a channel matrix
$\mathbf{H}$, where the integration is conducted on the faded
Voronoi cells $\mathcal{V}_{\mathcal{S}_{k,p}}(\mathbf{H})$. Thus,
by averaging these expressions over all fading realizations, the
average SEP is obtained as
\begin{equation}\label{inf_sep_h}
P_\infty(\rho)=1-\mathbb{E}\left[J_N(\mathbf{H})\right],
\end{equation}
and
\begin{equation}\label{KPAM_sep_h}
P_{K-PAM}(\rho)=1-\mathbb{E}\left[\frac{\sum\limits_{k=0}^N(K-1)^k\sum\limits_{p=1}^{\binom{N}{k}}J_{k,p}(\mathbf{H})}{K^N}\right],
\end{equation}
where
\begin{equation}\label{jkm_h}
J_{k,p}(\mathbf{H})=\int_{\mathcal{V}_{\mathcal{S}_{k,p}}(\mathbf{H})}
p(\mathbf{z})\mathrm{d}\mathbf{z},\,\,\,\,\,\,0<k<n,
\end{equation}
with $J_0(\mathbf{H})=1$, and $\mathbb{E}[\cdot]$ denotes
expectation with respect to the fading distribution. Moreover,
based on (\ref{inf_sep_h}) and (\ref{KPAM_sep_h}), the FEP can be
calculated by
\begin{equation}\label{fep_inf}
P_{f,\infty}(\rho)=1-\mathbb{E}\left[\left(J_N(\mathbf{H})\right)^L\right],
\end{equation}
and
\begin{equation}\label{fep_kpam}
P_{f,K-PAM}(\rho)=1-\mathbb{E}\left[\left(\frac{\sum\limits_{k=0}^N(K-1)^k\sum\limits_{p=1}^{\binom{N}{k}}J_{k,p}(\mathbf{H})}{K^N}\right)^L\right],
\end{equation}
for an infinite and a finite lattice constellation respectively.
To the best of the authors' knowledge, an expression for the FEP
of multidimensional lattice constellations as (\ref{fep_kpam}) has
not been previously given. The above expressions are difficult to
evaluate, due to the unknown shape of the faded Voronoi cells.
Therefore, in the following we provide upper and lower bounds for
these expressions.

\subsection{Bounds}\label{bounds}
Based on the exact expressions (\ref{fep_inf}) and
(\ref{fep_kpam}), we can now present lower and upper bounds for
the performance of infinite and finite lattice constellations.

\subsubsection{Lower Bounds}\label{lb}

For the readers' convenience, a well known lower bound  for
infinite lattice constellations which was investigated in
\cite{Fabregas1}, is revisited here. In this bound, the integral
on the faded Voronoi cell
$\mathcal{V}_{\mathcal{S}_N}(\mathbf{H})$ in (\ref{inf_sep_h}) and
(\ref{fep_inf}) is substituted by an integral on an
$N$-dimensional sphere of the same volume,
$\mathcal{B}_N(\mathbf{H})$, for which holds
\begin{equation}\label{vol_b_N}
\mathrm{vol}_N(\mathcal{B}_N(\mathbf{H}))=\mathrm{vol}_N(\mathcal{V}_{\mathcal{S}_N}(\mathbf{H}))=|\mathrm{det}(\mathbf{HM})|=\prod\limits_{i=1}^{N}h_i.
\end{equation}
However, the volume of each
$\mathcal{V}_{\mathcal{S}_{k,p}}(\mathbf{H})$ in (\ref{fep_kpam})
cannot be directly substituted in the same manner by an equality
such as in (\ref{vol_b_N}).

\begin{definition}\label{Rk_H}
We define the $k$-dimensional spheres $\mathcal{B}_k(\mathbf{H})$,
the radius $R_k(\mathbf{H})$ of which is given by
\begin{equation}\label{radius}
R_k^2(\mathbf{H})=\left\{\begin{array}{cc}
\frac{1}{\pi}\Gamma(\frac{k}{2}+1)^\frac{2}{k}W^2\max\limits_{i=1,\ldots,N}\gamma_i,&k=1,2,\ldots,(N-1)\\
\frac{1}{\pi}\Gamma(\frac{k}{2}+1)^\frac{2}{k}\left(\prod\limits_{i=1}^N\gamma_i\right)^\frac{1}{N},&k=N\\
\end{array}\right.
\end{equation}
where $\max\limits_{i=1,\ldots,N}\gamma_i$ is the maximum between
all $\gamma_i=h_i^2$ and

\begin{equation}\label{W}
W=\frac{\|\mathbf{v_1}\|+\|\mathbf{v_2}\|+\ldots+\|\mathbf{v_N}\|}{N},
\end{equation}
with $\|\mathbf{v_i}\|$ being the norm of vector $\mathbf{v_i}$.
Note that for $\mathbb{Z}^N$ lattices, $W=1$. For $k=N$, the
sphere $\mathcal{B}_N$ with radius $R_N(\mathbf{H})$ is of the
same volume as the Voronoi cell
$\mathcal{V}_{\mathcal{S}_{k,p}}(\mathbf{H})$, as in
\cite{Fabregas1}.
\end{definition}

\begin{definition}\label{Ik_H}
We define the integrals \cite{Fabregas1}
\begin{equation}\label{Ik1}
I_k(\mathbf{H})=\int\limits_{\mathcal{B}_k(\mathbf{H})}
p(\mathbf{z})\mathrm{d}\mathbf{z}=\left\{
\begin{array}{cc}
1,&k=0,\\
1-\frac{\Gamma\left(\frac{k}{2},\frac{R_k^2(\mathbf{H})}{2}\rho\right)}{\Gamma\left(\frac{k}{2}\right)},&k=1,2,\ldots,N,\\
\end{array}\right.
\end{equation}
where $\mathcal{B}_k(\mathbf{H})$ is defined in Definition
\ref{Rk_H}. When $k=0$, we define
$I_0(\mathbf{H})=J_0(\mathbf{H})=1$.
\end{definition}

The FEP of an infinite lattice constellation, given in
(\ref{fep_inf}), is lower-bounded by the following Sphere Lower
Bound (SLB) \cite{Fabregas1}

\begin{equation}\label{SLB_theorem}
P_{slb}\left( \rho \right) =1-\mathbb{E}\left[ \left(
I_N(\mathbf{H})\right) ^{L}\right]=1-\mathbb{E}\left[ \left(
1-\frac{\Gamma \left(
\frac{N}{2},\frac{R_{N}^{2}\left( \mathbf{H}\right) }{2}\rho \right) }{%
\Gamma \left( \frac{N}{2}\right) }\right) ^{L}\right].
\end{equation}

\begin{theorem}\label{MSLB}

The FEP of a multidimensional finite lattice constellation, given
in (\ref{fep_kpam}), is lower bounded by

\begin{equation}\label{bound}
P_{mslb}(\rho)=1-\mathbb{E}\left[\left(\frac{\sum\limits_{k=0}^N{(K-1)^k\binom{N}{k}I_k(\mathbf{H})}}{K^N}\right)^L\right],
\end{equation}
where $P_{mslb}(\rho)$ is called Multiple Sphere Lower Bound
(MSLB).
\end{theorem}

\begin{proof}
The proof is given in Appendix \ref{appendix}.
\end{proof}

\subsubsection{Upper Bounds}\label{ub}

The error performance of infinite lattice constellations in AWGN
channels is upper bounded by the well-known upper Sphere Upper
Bound (SUB), presented in \cite{Pappi}. Similarly, a Multiple
Sphere Upper Bound (MSUB) is also proposed in \cite{Pappi} for
finite lattice constellations. These bounds are based on the
minimum distance between any two points of the lattice.

\begin{definition}\label{R_sub}
We define the $k$-dimensional spheres $\mathcal{G}_k(\mathbf{H})$,
the radius of which is given by
\begin{equation}\label{radius_sub}
\mathcal{R}^2(\mathbf{H})=\left(\frac{d_{min}}{2}\min\limits_{i=1,\ldots,N}h_i\right)^2=\frac{d_{min}^2}{4}\min\limits_{i=1,\ldots,N}\gamma_i.
\end{equation}
\end{definition}

\begin{definition}\label{Isub_H}
We define the integrals
\begin{equation}\label{Iksub1}
\mathcal{I}_k(\mathbf{H})=\int\limits_{\mathcal{G}_k(\mathbf{H})}
p(\mathbf{z})\mathrm{d}\mathbf{z}=\left\{\begin{array}{cc}
1,&k=0,\\
1-\frac{\Gamma\left(\frac{k}{2},\frac{\mathcal{R}^2(\mathbf{H})}{2}\rho\right)}{\Gamma\left(\frac{k}{2}\right)},&k=1,2,\ldots,N,\\
\end{array}\right.
\end{equation}
where $\mathcal{G}_k(\mathbf{H})$ is a $k$-dimensional sphere,
with radius defined in (\ref{radius_sub}). When $k=0$, we define
$\mathcal{I}_0(\mathbf{H})=J_0(\mathbf{H})=1$.
\end{definition}

\begin{theorem}\label{sub_theorem}
The FEP of an multidimensional infinite lattice constellation is
upper bounded by
\begin{equation}\label{sub}
P_{sub}(\rho)=1-\mathbb{E}\left[\left(\mathcal{I}_N(\mathbf{H})\right)^L\right]=1-\mathbb{E}\left[\left(1-\frac{\Gamma\left(\frac{N}{2},\frac{\mathcal{R}^2(\mathbf{H})}{2}\rho\right)}{\Gamma\left(\frac{N}{2}\right)}\right)^L\right],
\end{equation}
where $P_{sub}(\rho)$ is called Sphere Upper Bound (SUB).
\end{theorem}
\begin{proof}
The proof is given in Appendix \ref{appendixb}.
\end{proof}

\begin{theorem}\label{msub_theorem}
The FEP of a multidimensional finite lattice constellation is
upper bounded by
\begin{equation}\label{boundmsub}
P_{msub}(\rho)=1-\mathbb{E}\left[\left(\frac{\sum\limits_{k=0}^N{(K-1)^k\binom{N}{k}\mathcal{I}_k}(\mathbf{H})}{K^N}\right)^L\right],
\end{equation}
where $P_{msub}(\rho)$ is called Multiple Sphere Upper Bound
(MSUB).
\end{theorem}
\begin{proof}
The proof is given in Appendix \ref{appendixc}.
\end{proof}

\subsection{Closed-Form Analysis}\label{closed_form}
Next, we define three functions which will be used in deriving
closed-form expressions for the bounds presented above.

\begin{definition}\label{function_A}
We define the function
\begin{equation}\label{f_A}
\begin{array}{cc} A(\rho,N;k,L)=\mathbb{E}\left[ \left( 1-\frac{\Gamma
\left(
\frac{k}{2},\frac{\mathcal{R}^{2}\left( \mathbf{H}\right) }{2}\rho \right) }{%
\Gamma \left( \frac{k}{2}\right) }\right) ^{L}\right] ,&0<k\leq N,
\end{array}
\end{equation}
where $\mathcal{R}$ is given in (\ref{radius_sub}).
\end{definition}
\begin{proposition}
The above function $A(\rho,N;k,L)$, when $N$ is even, can be
written in closed-form as
\begin{equation}\label{cl_f_A}
A(\rho,N;k,L) =1+\sum\limits_{q=1}^{L}\sum_{\substack{ %
n_{0},...,n_{\frac{k}{2}-1}=0 \\
n_{0}+...+n_{\frac{k}{2}-1}=q}}^{q}\frac{\left( -1\right) ^{q}L!
Nm^{m}}{\Psi\left( L-q\right) !\Gamma \left( m\right) }
\sum_{\substack{ %
t_{0},...,t_{m-1}=0 \\
t_{0}+...+t_{m-1}=N-1}}^{N-1}\frac{m^{\mathcal{Y}}\left( \frac{d_{\min }^{2}\rho }{%
8}\right) ^{\mathcal{Z}}\Gamma \left(
\mathcal{Y}+m+\mathcal{Z}\right) }{ \Xi\left( mN+\frac{q\rho
d_{\min }^{2}}{8}\right) ^{\mathcal{Y}+m+\mathcal{Z}}}.
\end{equation}
where $\mathcal{Y}=\sum\limits_{j=0}^{m-1}jt_{j}$,
$\mathcal{Z}=\sum\limits_{i=0}^{\frac{k}{2}-1}in_{i}$,
$\Xi=\prod\limits_{j=0}^{m-1}\left( \left( j!\right)
^{t_{j}}\Gamma \left( t_{j}+1\right)\right)$ and
$\Psi=\prod\limits_{i=0}^{\frac{k}{2}-1}\left( \left(
i!\right)^{n_{i}}\Gamma \left( n_{i}+1\right)\right)$.
\end{proposition}
\begin{proof}
The proof is given in Appendix \ref{appendixd}.
\end{proof}

\begin{definition}\label{function_B}
We define the function
\begin{equation}\label{f_B}
\begin{array}{cc} B(\rho,N;k)=\mathbb{E}\left[  1-\frac{\Gamma
\left(
\frac{k}{2},\frac{R_k^{2}\left( \mathbf{H}\right) }{2}\rho \right) }{%
\Gamma \left( \frac{k}{2}\right) }\right] ,&0<k<N,
\end{array}
\end{equation}
where $R_k$ is given in (\ref{radius}) for $k\neq N$.
\end{definition}
\begin{proposition}
The above function $B(\rho,N;k)$, when $N$ is even, can be written
in closed-form as
\begin{equation}\label{cl_f_b}
\begin{array}{ll}
B(\rho,N;k) =1-\sum\limits_{q=1}^{N}\frac{\binom{N}{q}\Gamma \left( q+1\right)}{\Gamma \left( \frac{k}{2}\right)}&\times \left\{\sum\limits_{\substack{ _{\substack{ n_{0},n_{1},...,n_{m-1}=0 \\ %
n_{0}+n_{1}+...+n_{m-1}=q}} \\ \mathcal{X}\neq 0}}%
^{q}\mathcal{X}\Upsilon g\left(\mathcal{X},%
\frac{\rho \Gamma \left( \frac{k}{2}+1\right) ^{\frac{2}{k}}W^{2}}{2\pi },qm,%
\frac{k}{2}\right)\right.\\
&\left.- \sum\limits_{\substack{ n_{0},n_{1},...,n_{m-1}=0 \\ %
n_{0}+n_{1}+...+n_{m-1}=q}}^{q} qm\Upsilon g\left(
\mathcal{X}+1,\frac{\rho \Gamma \left( \frac{k}{2}%
+1\right) ^{\frac{2}{k}}W^{2}}{2\pi
},qm,\frac{k}{2}\right)\right\}
\end{array}
\end{equation}
with $\mathcal{X}=\sum\limits_{i=0}^{m-1}in_{i}$, $\Upsilon=\prod\limits_{i=0}^{m-1}\frac{\left( \frac{m^{i}}{i!}\right) ^{n_{i}}}{%
\Gamma \left( n_{i}+1\right) }$ and
\begin{equation}\label{func_g}
g\left( \alpha ,\beta ,p,\nu \right) =-\frac{\beta ^{\nu }\Gamma
\left( \alpha +\nu \right) }{\nu p^{\alpha +\nu }}~_{2}F_{1}\left(
\nu ,\alpha +\nu ;\nu +1;-\frac{\beta }{p}\right) +\frac{\Gamma
\left( \nu \right) \Gamma \left( \alpha \right) }{p^{\alpha }}.
\end{equation}
In (\ref{func_g}), $~_{2}F_{1}\left( \alpha ,\beta;\gamma ;
z\right) $ is the Gauss Hypergeometric function defined by
\cite[Eq. (9.100),(9.14)]{Gradshteyn}.
\end{proposition}
\begin{proof}
The proof is given in Appendix \ref{appendixe}.
\end{proof}

\begin{definition}\label{function_C}
We define the function
\begin{equation}\label{f_C}
C(\rho,N;L)=\mathbb{E}\left[ \left( 1-\frac{\Gamma \left(
\frac{N}{2},\frac{R_N^{2}\left( \mathbf{H}\right) }{2}\rho \right) }{%
\Gamma \left( \frac{N}{2}\right) }\right)^L\right] ,
\end{equation}
where $R_N$ is given in (\ref{radius}) for $k=N$.
\end{definition}
\begin{proposition}
The above function $C(\rho,N;L)$, when $N$ is even, can be written
in closed-form as
\begin{equation}\label{slb3}
\begin{split}
C(\rho,N;L)  &=1+\sum_{q=1}^{L}\sum_{\substack{ n_{0},...,n_{%
\frac{N}{2}-1}=0  \\
n_{0}+...+n_{\frac{N}{2}-1}=q}}^{q}\frac{L!\left( -1\right)
^{q}\sqrt{N}\left( \frac{N}{q}\right) ^{\mathcal{Q}}}{\left(
L-q\right)
!\left( \Gamma \left( m\right) \right) ^{N}\left( 2\pi \right) ^{\frac{N-1}{2}%
}\Phi } \\
&\times G_{N,N}^{N,N}\left[ \left( \frac{\frac{2\pi
mN}{q\rho }}{\left( \Gamma \left( \frac{N}{2}+1\right) \right)
^{\frac{2}{N}}}\right) ^{N}\left\vert
\begin{array}{c}
\frac{1-\mathcal{Q}}{N},...,\frac{%
N-\mathcal{Q}}{N} \\
m,...,m \\
\end{array}%
\right. \right],
\end{split}
\end{equation}
where $\mathcal{Q}=\sum\limits_{i=0}^{\frac{N}{2}-1}in_{i}$,
$\Phi=\prod\limits_{i=0}^{\frac{N}{2}-1}\left( \left( i!\right)
^{n_{i}}\Gamma \left( n_{i}+1\right) \right)$ and
$G_{p,q}^{m,n}\left[ \cdot \right] $ is the Meijer's G-function
\cite[Eq. (9.301)]{Gradshteyn}.
\end{proposition}

\begin{proof}
The proof is given in Appendix \ref{appendixf}.
\end{proof}

Note that for the important case of $N=2$, $C(\rho,N;L)$ can be
written in terms of the more familiar Gauss Hypergeometric
function, $~_{2}F_{1}\left(\cdot,\cdot;\cdot;\cdot\right)$
as\cite{wolfram}
\begin{equation}
C(\rho,2;L)=1+\sum_{q=1}^{L}\frac{\left( -1\right)^{q}L!
\left(\Gamma\left(\frac{1}{2}+m\right)\right)^2\left( \frac{4\pi
m}{q\rho } \right) ^{2m}}{\sqrt{\pi}\left( L-q\right)!\Gamma
\left(
q+1\right)\Gamma\left(\frac{1}{2}+2m\right)}~_{2}F_{1}\left(\frac{1}{2}+m,m;\frac{1}{2}+2m;1-\left(
\frac{4\pi m}{q\rho } \right) ^{2}\right).
\end{equation}

\subsubsection{Closed-Form for the SLB}\label{cl_slb}
The SLB for the FEP of infinite lattice constellations of even
dimension $N$ is given in closed-form by
\begin{equation}\label{slb_final}
P_{slb}(\rho)=1-\mathbb{E}\left[ \left( I_N(\mathbf{H})\right)
^{L}\right]=1-C(\rho,N;L).
\end{equation}
\subsubsection{Closed-Form for the SUB}\label{cl_sub}
The SUB for the FEP of infinite lattice constellations of even
dimension $N$ is given in closed-form by
\begin{equation}\label{sub_final}
P_{sub}(\rho)=1-\mathbb{E}\left[ \left(
\mathcal{I}_N(\mathbf{H})\right) ^{L}\right]=1-A(\rho,N;N,L).
\end{equation}
\subsubsection{Closed-Form for the MSLB}\label{cl_mslb}
The MSLB for the SEP of infinite lattice constellations of even
dimension $N$ in fading channels ($L=1$) is given in closed-form
by
\begin{equation}\label{mslb_final}
\begin{split}
P_{mslb}(\rho)&=1-\mathbb{E}\left[\frac{\sum\limits_{k=0}^N{(K-1)^k\binom{N}{k}I_k(\mathbf{H})}}{K^N}\right]\\
&=1-\frac{1+\sum\limits_{k=1}^{N-1}\left[(K-1)^k\binom{N}{k}B(\rho,N;k)\right]+(K-1)^NC(\rho,N;1)}{K^N}.
\end{split}
\end{equation}
\subsubsection{Closed-Form for the MSUB}\label{cl_msub}
The MSUB for the SEP of infinite lattice constellations of even
dimension $N$ in fading channels ($L=1$) is given in closed-form
by
\begin{equation}\label{msub_final}
P_{msub}(\rho)=1-\mathbb{E}\left[\frac{\sum\limits_{k=0}^N{(K-1)^k\binom{N}{k}\mathcal{I}_k(\mathbf{H})}}{K^N}\right]=1-\frac{1+\sum\limits_{k=1}^{N}(K-1)^k\binom{N}{k}A(\rho,N;k,1)}{K^N}.
\end{equation}

\section{Numerical \& Simulation Results}\label{results}

In this section, we investigate the accuracy and tightness of the
proposed bounds, which are compared to the performance of
$\mathbb{Z}^N$ infinite lattice constellations and various finite
lattice constellations. The $\mathbb{Z}^N$ lattices are the mostly
used in practical applications, since the bit labeling of
constellations carved from them is straightforward and Gray coding
can be implemented in most cases. In the following, we consider
two different types of $\mathbb{Z}^N$ lattices, those which are
optimally rotated in terms of full diversity and maximization of
the minimum product distance and non rotated lattices which
perform poorly, due to low diversity gain \cite{Viterbo2,
Viterbo3}. Note that the performance of any other rotation of
these lattices falls between the performance of these two extreme
cases examined. For both types of lattices, the normalization of
the generator matrix results in $d_{min}=1$ and $W=1$.

Fig. \ref{Fig:inf_varL} depicts the accuracy and tightness of the
SLB and the SUB along with the frame error probability of
$\mathbb{Z}^2$ infinite lattice constellations, for various values
of frame lengths. Specifically, the analytical results obtained
from (\ref{slb_final}) and (\ref{sub_final}) are plotted in
conjunction with simulation results for the frame error
probability of the cyclotomic rotation of the $\mathbb{Z}^2$
infinite lattice constellation and the non rotated lattice, when
$m=1$ and $L=1,100$. As it is clearly illustrated, the numerical
results obtained from  the analytical expressions act as lower and
upper bounds in all cases examined. In particular, as the known
SLB is observed to be very close to the performance of optimally
rotated lattices, the SUB seems to be less tight but still very
close to the non-rotated case. Moreover, as the frame length $L$
increases, it is evident that the proposed SUB becomes even
tighter. It can be observed that the diversity order, i.e. the
asymptotic slope of the frame error probability, is independent of
the frame length, for both the lattices under investigation and
their SLB and SUB.

Fig. \ref{Fig:inf_varm} demonstrates the results for SLB and SUB
along with the simulated performance of the $\mathbb{Z}^2$
infinite lattice constellations for various values of the \emph{m}
parameter. It is evident that both the SLB and the SUB act as
tight bounds, irrespective of $m$. Moreover, the effects of the
$m$-parameter on the diversity order of both lattices under
investigation and their bounds, are clearly depicted.

Fig. \ref{Fig:inf_varN} illustrates the effect of the dimension
order on the frame error probability of the $\mathbb{Z}^N$
infinite lattice constellations and the corresponding SLB and SUB.
In particular, we consider $m=1$, $L=1$ and $N=2,8$. It is obvious
that the SLB and SUB act as bounds irrespective of the dimension
$N$, while the SUB is tighter for small dimension. The SUB has
similar diversity order with the non rotated lattices and the SLB
has similar diversity order with the optimally rotated lattices
which achieve full diversity.

Figs. \ref{Fig:fin_varL}, \ref{Fig:fin_varm} and
\ref{Fig:fin_varN} illustrate the performance obtained via
Monte-Carlo simulation of finite $\mathbb{Z}^N$ $4$-PAM lattice
constellations for various values of block length $L$, parameter
$m$ and dimension $N$ respectively. These figures correspond to
the cases of Figs. \ref{Fig:inf_varL}, \ref{Fig:inf_varm} and
\ref{Fig:inf_varN} for infinite lattice constellations. The MSLB
and MSUB are illustrated for each case, in addition to the
corresponding SLB and SUB for the infinite lattices as a
reference. One can observe that the behavior of the MSLB and MSUB
with respect to the performance of a finite constellation is
extremely similar to that of the SLB and SUB with respect to the
performance of an infinite constellation. However, it is clearly
illustrated that the MSLB and MSUB bounds are more appropriate for
a finite lattice constellations than the SLB and SUB.
Specifically, as shown in Fig. \ref{Fig:fin_varL}, the SLB does
not act as a bound for finite constellations, whereas the MSLB is
always a lower bound and it is tighter to the simulated
performance of the optimally rotated lattices than the SLB.
Similarly, the MSUB is tighter than the SUB to the simulated
performance of the non-rotated lattices. For higher values of the
$m$ parameter, while the SLB is not a lower bound for low SNR
values, the MSLB remains below the simulated performance of the
optimally rotated lattices for all values of SNR. In Fig.
\ref{Fig:fin_varN}, the MSLB for higher dimensions is rather loose
for low SNR values, but becomes tighter as the SNR increases.
However, it is the only reliable bound, since the SLB is above the
simulated performance of the optimally rotated lattice for all SNR
values. Both expressions can be used, the MSLB as a reliable lower
bound and the SLB as a good approximation of the actual
performance. Finally for the MSUB, it is tighter than the SUB in
all cases.

The MSLB and MSUB also take into account the number $K$ of points
along the direction of each basis vector of the constellations. In
Fig. \ref{Fig:fin_K32}, constellations of larger $K$ are depicted,
that is the optimally rotated and the non-rotated $\mathbb{Z}^2$
$32$-PAM. It is again evident that the MSLB is an extremely tight
lower bound, regardless of the rank of the constellation.
Moreover, the MSUB is also accurate and tight, illustrating that
the tightness of the bounds is not noticeably affected by the rank
of the constellations. Finally, it can be deduced that as the
parameter $K$ increases, the MSLB and MSUB tend to coincide with
the SLB and the SUB respectively. This is expected, since the
number of inner points on the constellation, which are
approximated in the same way in the two pairs of bounds, is much
larger than the number of outer points.

Finally, in Fig. \ref{Fig:fin_A2}, the performance of a
constellation carved from a lattice with different structure is
depicted, together with the numerical results for the
corresponding MSLB and MSUB. Specifically, the $\mathbb{A}^2$
$4$-PAM constellation is illustrated, which is the best known
packing lattice in two dimensions \cite{Conway}, with generator
matrix
\begin{equation}
M=\left[
\begin{array}{cc}
\sqrt{\frac{2}{\sqrt{3}}}&\sqrt{\frac{1}{2\sqrt{3}}}\\
0&\sqrt{\frac{3}{2\sqrt{3}}}
\end{array}\right].
\end{equation}
The MSLB and MSUB again bound the performance of this lattice
constellation, while the SLB is not a reliable bound and the SUB
is looser than the MSUB. Moreover, the figure suggests that the
simulated rotation of the constellation is not optimal with
respect to the diversity gain and the maximization of the minimum
product distance. This is an important result, not only for this
constellation but also for other signal sets with random rotation,
because the bounds act as indication of the optimality of a
designed constellation.

\section{Conclusions}\label{conclusions}
We have studied the error performance of multidimensional lattice
constellations in block fading channels. We first presented
analytical expressions for the exact FEP of both infinite and
finite signal sets carved from lattices, in the presence of
Nakagami-\emph{m} block fading. These expressions were then
bounded by the well known SLB and the novel SUB proposed for the
infinite lattice constellations, as well as by the proposed MSLB
and MSUB for finite lattice constellations of arbitrary structure,
rank and dimension. Then, analytical closed form expressions were
derived for the SUB and SLB in lattices with even dimensions,
whereas the MSLB and MSUB were given in closed form for
constellations of even dimensions and for transmission in fading
channels with single-symbol block length. The proposed analytical
framework sets the performance limits of such signal sets and it
can be an efficient tool for their analysis and design.

\appendices

\section{Proof of Theorem \ref{MSLB}}\label{appendix}

The volume of $\mathcal{V}_{\mathcal{S}_{k,p}}(\mathbf{H})$ in
(\ref{jkm_h}) is the same as the volume of the corresponding
fundamental parallelotope of a faded sublattice, for which holds
\cite[Eq. (41)]{Pappi}

\begin{equation}\label{voronoivol}
\mathrm{vol}_k(\mathcal{V}_{\mathcal{S}_{k,p}}(\mathbf{H}))\leq\prod_{i:\mathbf{v}_i\in\mathcal{S}_{k,p}}\|\mathbf{Hv}_i\|\leq\prod_{i:\mathbf{v}_i\in\mathcal{S}_{k,p}}\|\mathbf{v}_i\|\max\limits_{j=1,\ldots,N}h_j,
\end{equation}
where the first equality is true only when the vectors of
$\mathcal{S}_{k,p}$ distorted by fading are orthogonal. The second
equality holds only if all fading coefficients $h_j$ are equal to
$\max\limits_{j=1,\ldots,N}h_j$. Moreover, the inequalities in
(\ref{voronoivol}) imply that the volume of a faded Voronoi cell
of a sublattice $\mathcal{S}_{k,p}$ will always be at most equal
to the volume of a rectangular parallelotope, with edges of norms
equal to those of the vectors in $\mathcal{S}_{k,p}$, multiplied
by the largest fading coefficient. Using (\ref{voronoivol}) yields
to

\begin{equation}\label{voronineq}
\sum\limits_{p=1}^{\binom{N}{k}}\mathrm{vol}_k(\mathcal{V}_{\mathcal{S}_{k,p}}(\mathbf{H}))\leq\sum\limits_{p=1}^{\binom{N}{k}}\prod_{i:\mathbf{v}_i\in\mathcal{S}_{k,p}}\|\mathbf{v}_i\|\max\limits_{j=1,\ldots,N}h_j,
\end{equation}
which can be written as
\begin{equation}\label{voronineq2}
\sum\limits_{p=1}^{\binom{N}{k}}\mathrm{vol}_k(\mathcal{V}_{\mathcal{S}_{k,p}}(\mathbf{H}))\leq\left(\max\limits_{j=1,\ldots,N}h_j\right)^k\sum\limits_{\substack{b_1+b_2+\ldots+b_N=k\\b_1,b_2,\ldots,b_N\in\{0,1\}}}\|\mathbf{v}_1\|^{b_1}\|\mathbf{v}_2\|^{b_2}\cdots\|\mathbf{v}_N\|^{b_N}.
\end{equation}

Using Maclaurin's Inequality \cite[p.52]{Hardy}, for
$a_1,a_2,\ldots,a_N\in\mathbb{R}$ and $0<k<N$,

\begin{equation}\label{Maclaurin}
\varrho_N^{\frac{1}{N}}\leq\varrho_k^{\frac{1}{k}}\leq\varrho_1,
\end{equation}
where
\begin{equation}\label{varrho}
\varrho_k=\frac{\sum\limits_{\substack{b_1+b_2+\ldots+b_N=k\\b_1,b_2,\ldots,b_N\in\{0,1\}}}a_1^{b_1}a_2^{b_2}\cdots
a_N^{b_N}}{\binom{N}{k}}.
\end{equation}
If we set $a_i=\|\mathbf{v}_i\|$, $i=1,2,\ldots,N$, then
$\varrho_1=W$ and from (\ref{Maclaurin}) and (\ref{varrho})

\begin{equation}\label{normineq}
\sum\limits_{\substack{b_1+b_2+\ldots+b_N=k\\b_1,b_2,\ldots,b_N\in\{0,1\}}}\|\mathbf{v}_1\|^{b_1}\|\mathbf{v}_2\|^{b_2}\cdots\|\mathbf{v}_N\|^{b_N}\leq\binom{N}{k}W^k.
\end{equation}
From (\ref{voronineq2}) and (\ref{normineq}), for $0<k<N$, we have

\begin{equation}\label{sphineq}
\sum\limits_{p=1}^{\binom{N}{k}}\mathrm{vol}_k(\mathcal{V}_{\mathcal{S}_{k,p}}(\mathbf{H}))\leq\left(\max\limits_{j=1,\ldots,N}h_j\right)^k\binom{N}{k}W^k.
\end{equation}
For the spheres $\mathcal{B}_{\mathcal{S}_{k,p}}(\mathbf{H})$ with
$\mathrm{vol}_k(\mathcal{B}_{\mathcal{S}_{k,p}}(\mathbf{H}))=\mathrm{vol}_k(\mathcal{V}_{\mathcal{S}_{k,p}}(\mathbf{H}))$,
holds that

\begin{equation}\label{sumsmallspheres}
\sum\limits_{p=1}^{\binom{N}{k}}J_{k,p}(\mathbf{H})\leq\sum\limits_{p=1}^{\binom{N}{k}}
\int_{\mathcal{B}_{\mathcal{S}_{k,p}}(\mathbf{H})}
p(\mathbf{z})\mathrm{d}\mathbf{z}=\sum\limits_{p=1}^{\binom{N}{k}}\left(1-\frac{\Gamma\left(\frac{k}{2},\frac{R_{\mathcal{S}_{k,p}}^2(\mathbf{H})}{2}\rho\right)}{\Gamma\left(\frac{k}{2}\right)}\right),
\end{equation}
where $R_{\mathcal{S}_{k,p}}(\mathbf{H})$ is the radius of
$\mathcal{B}_{\mathcal{S}_{k,p}}(\mathbf{H})$. From
(\ref{sphineq}), and using that
$\mathrm{vol}_k\left(\mathcal{B}_{\mathcal{S}_{k,p}}(\mathbf{H})\right)=\frac{\pi^{\frac{k}{2}}R_{\mathcal{S}_{k,p}}^k(\mathbf{H})}{\Gamma\left(\frac{k}{2}+1\right)}$

\begin{equation}\label{radiusineq1}
\sum\limits_{p=1}^{\binom{N}{k}}\frac{\pi^{\frac{k}{2}}R_{\mathcal{S}_{k,p}}^k(\mathbf{H})}{\Gamma\left(\frac{k}{2}+1\right)}\leq\left(\max\limits_{j=1,\ldots,N}h_j\right)^k\binom{N}{k}W^k.
\end{equation}
Furthermore, by taking into account (\ref{radius}) for the case
when $0<k<N$,

\begin{equation}\label{radiusineq2}
\sum\limits_{m=1}^{\binom{N}{k}}R_{\mathcal{S}_{k,m}}^k(\mathbf{H})\leq\binom{N}{k}R_k^k(\mathbf{H}).
\end{equation}

As proved in \cite{Pappi}, the function $f(x;a,b)=\Gamma \left(
a,b x^{1/a} \right)$ is convex in $(0,\infty)$, thus from Jensen's
Inequality for convex functions \cite{Hardy} holds that
\begin{equation}\label{jensen}
\sum_{i=1}^M {\Gamma \left( a,b {x_i}^{1/a} \right)} \geq M\Gamma
\left( a,b \left( \sum_{i=1}^L {x_i} /L \right) ^{1/a} \right).
\end{equation}
For $a=\frac{k}{2}$ , $b=\frac{\rho}{2}$, $M=\binom{N}{k}$ and
$x_i={R^k_{\mathcal{S}_{k,p}}}$ we get

\begin{equation}\label{gammaineq1}
\sum_{p=1}^{\binom{N}{k}} {\Gamma \left(
\frac{k}{2},\frac{\rho}{2} R^2_{\mathcal{S}_{k,p}}(\mathbf{H})
\right)} \geq \binom{N}{k}\Gamma \left( \frac{k}{2},\frac{\rho}{2}
\left( \frac{\sum\limits_{m=1}^{\binom{N}{k}}
R^k_{\mathcal{S}_{k,p}}(\mathbf{H})}{\binom{N}{k}} \right)
^{\frac{2}{k}} \right).
\end{equation}
From (\ref{radiusineq2}) and since $f(x;a,b)=\Gamma \left( a,b
x^{1/a} \right)$ is a monotonically decreasing function with
respect to $x$,

\begin{equation}\label{gammaineq2}
\Gamma \left( \frac{k}{2},\frac{\rho}{2} \left(
\frac{\sum\limits_{p=1}^{\binom{N}{k}}
R^k_{\mathcal{S}_{k,p}}(\mathbf{H})}{\binom{N}{k}} \right)
^{\frac{2}{k}} \right)\geq\Gamma \left(
\frac{k}{2},\frac{\rho}{2}R^2_k(\mathbf{H}) \right).
\end{equation}
Using (\ref{gammaineq1}) and (\ref{gammaineq2}), for $0<k<N$

\begin{equation}\label{gammaineq3}
\sum_{p=1}^{\binom{N}{k}} {\Gamma \left(
\frac{k}{2},\frac{\rho}{2} R^2_{\mathcal{S}_{k,p}}(\mathbf{H})
\right)} \geq\binom{N}{k}\Gamma \left( \frac{k}{2},\frac{\rho}{2}
R^2_k(\mathbf{H}) \right),
\end{equation}
or equivalently

\begin{equation}\label{gammaineq4}
\sum_{p=1}^{\binom{N}{k}} \left(1-\frac{\Gamma \left(
\frac{k}{2},\frac{\rho}{2} R^2_{\mathcal{S}_{k,p}}(\mathbf{H})
\right)}{\Gamma\left(\frac{k}{2}\right)}\right)
\leq\binom{N}{k}\left(1-\frac{\Gamma \left(
\frac{k}{2},\frac{\rho}{2} R^2_k(\mathbf{H})
\right)}{\Gamma\left(\frac{k}{2}\right)}\right).
\end{equation}
Taking into account (\ref{sumsmallspheres}) and (\ref{gammaineq4})
for some $k$, $0<k<N$,  it yields

\begin{equation}\label{ineq_k}
\sum\limits_{p=1}^{\binom{N}{k}}J_{k,p}\leq\binom{N}{k}\left(1-\frac{\Gamma
\left( \frac{k}{2},\frac{\rho}{2} R^2_k(\mathbf{H})
\right)}{\Gamma\left(\frac{k}{2}\right)}\right)=\binom{N}{k}I_k(\mathbf{H}).
\end{equation}
For the case when $k=0$, $p=1$ and it holds that
$J_0(\mathbf{H})=I_0(\mathbf{H})=1$. For $k=N$, it is also $p=1$
and from (\ref{radius}), since
$\mathrm{vol}_N\left(\mathcal{V}_{\mathcal{S}_{N}}(\mathbf{H})\right)=\frac{\pi^{\frac{N}{2}}R_{N}^k(\mathbf{H})}{\Gamma\left(\frac{N}{2}+1\right)}$

\begin{equation}\label{ineq_N}
J_N(\mathbf{H})\leq\left(1-\frac{\Gamma \left(
\frac{N}{2},\frac{\rho}{2} R^2_N(\mathbf{H})
\right)}{\Gamma\left(\frac{N}{2}\right)}\right)=I_N(\mathbf{H}).
\end{equation}
Combining (\ref{ineq_k}) and (\ref{ineq_N}), multiplying by
$\left(K-1\right)^k$ and summing for all $k$, it results to

\begin{equation}\label{finalinequality}
\sum\limits_{k=0}^N\left(K-1\right)^k\sum\limits_{p=1}^{\binom{N}{k}}J_{k,p}(\mathbf{H})\leq\sum\limits_{k=0}^N\left(K-1\right)^k\binom{N}{k}I_k(\mathbf{H}).
\end{equation}
Finally, using (\ref{fep_kpam}), (\ref{bound}) and
(\ref{finalinequality})

\begin{equation}\label{proofslb}
P_{mslb}(\rho)\leq P_{f,K-PAM}(\rho)
\end{equation}
and this concludes the proof.

\section{Proof of Theorem \ref{sub_theorem}}\label{appendixb}

The proof starts by approximating the decision region of the faded lattice, $%
\mathcal{V}_{\mathcal{S}_N}\left( \mathbf{H}\right) $, with a
sphere, whose radius is equal with the packing radius of the
lattice \cite{Viterbo,Conway}, i.e., the minimum Euclidean
distance between the origin of the lattice and the facets of
$\mathcal{V}_{\mathcal{S}_N}\left( \mathbf{H}\right) $. If $q$ is
the number of neighboring symbols around a point of the unfaded
lattice, and $\mathbf{d_{i}}$ with $i=1,...,q$, is the vector from
the point investigated to the $i$-th neighboring one, the sphere
packing radius for a given channel realization $\mathbf{H}$
becomes equal to the minimum Euclidean distance on the faded
lattice, namely
\begin{equation}
d_{min,\mathcal{S}_N}(\mathbf{H})=\min\limits_{i=1,\ldots
,q}\frac{\Vert \mathbf{Hd_{i}}\Vert }{2} =\min\limits_{i=1,\ldots
,q} \frac{\sqrt{\sum\limits_{j=1}^{N}h_{j}^{2}d_{ij}^{2}}}{2}
\end{equation}
However for any $\mathbf{d_{i}}$ holds that%
\begin{equation}
\sqrt{\sum\limits_{j=1}^{N}h_{j}^{2}d_{ij}^{2}}\geq \Vert \mathbf{d_{i}}%
\Vert \min_{j=1,\ldots ,N}h_{j}.
\end{equation}%
Thus, we can conclude that
\begin{equation}
\min\limits_{i=1,\ldots ,q}\left( \frac{\sqrt{\sum%
\limits_{j=1}^{N}h_{j}^{2}d_{ij}^{2}}}{2}\right) \geq
\min\limits_{i=1,\ldots ,q}\frac{\Vert \mathbf{d_{i}}\Vert
\min\limits_{j=1,\ldots ,N}\left( h_{j}\right) }{2},
\end{equation}%
or equivalently
\begin{equation}
d_{min,\mathcal{S}_N}(\mathbf{H})=\min\limits_{i=1,\ldots ,q}\frac{\sqrt{\sum%
\limits_{j=1}^{N}h_{j}^{2}d_{ij}^{2}}}{2}\geq \frac{d_{\min }}{2}%
\min\limits_{j=1...,N}h_{j},
\end{equation}%
where $d_{min}=\min\limits_{i=1,\ldots ,q} \Vert
\mathbf{d_{i}}\Vert $ is the minimum Euclidean distance between
adjacent points on the unfaded infinite lattice constellation.
Therefore, the packing radius of the faded lattice can be lower
bounded by (\ref{radius_sub}), which yields an upper bound on the
frame error probability, $P_{sub}(\rho)$, given in (\ref{sub}).
This concludes the proof.

\section{Proof of Theorem \ref{msub_theorem}}\label{appendixc}

For a faded sublattice defined by $\mathcal{S}_{k,p}$, similarly
to Appendix \ref{appendixb}, the minimum Euclidean distance
$d_{min,\mathcal{S}_{k,p}}(\mathbf{H})$ between adjacent points
can be lower bounded by

\begin{equation}
d_{min,\mathcal{S}_{k,p}}(\mathbf{H})\geq
\frac{d_{min,\mathcal{S}_{k,p}}}{2}\min\limits_{j=1...,N} h_{j},
\end{equation}
where $d_{min,\mathcal{S}_{k,p}}$ is the minimum Euclidean
distance between adjacent points on the unfaded sublattice defined
by $\mathcal{S}_{k,p}$. Since
$\mathcal{S}_{k,p}\subseteq\mathcal{S}_N$, it holds that
\begin{equation}
d_{min,\mathcal{S}_{k,p}}\geq d_{min}
\end{equation}
and consequently
\begin{equation}
d_{min,\mathcal{S}_{k,p}}(\mathbf{H})\geq
\frac{d_{min}}{2}\min\limits_{j=1...,N}h_{j}.
\end{equation}
Thus, the packing radius of every faded sublattice defined by
$\mathcal{S}_{k,p}$ can be lower bounded by (\ref{radius_sub}) and
the integrals in (\ref{Iksub1}) are lower bounds to the integrals
on the faded Voronoi cells
$\mathcal{V}_{\mathcal{S}_{k,p}}(\mathbf{H})$, making the
expression in (\ref{boundmsub}) an upper bound of the FEP of
finite lattice constellations. This concludes the proof.

\section{Closed Form for the Function $A(\rho,N;k,L)$}\label{appendixd}
The following function

\begin{equation}\label{funcA}
A(\rho,N;k,L)=\mathbb{E}\left[ \left( 1-\frac{\Gamma \left(
\frac{k}{2},\frac{\mathcal{R}^{2}\left( \mathbf{H}\right) }{2}\rho \right) }{%
\Gamma \left( \frac{k}{2}\right) }\right) ^{L}\right] ,
\end{equation}
can be written as \cite[Eq. (1.111)]{Gradshteyn}

\begin{equation}\label{bound_u}
A(\rho,N;k,L)=\sum\limits_{q=0}^{L}\binom{L}{q}\frac{\left(
-1\right) ^{q}}{\left( \Gamma \left( \frac{k}{2}\right) \right)
^{q}} \mathbb{E}\left[ \left( \Gamma \left( \frac{k}{2},\frac{%
\mathcal{R}^{2}\left( \mathbf{H}\right) }{2}\rho \right) \right)
^{q}\right] .
\end{equation}%
Using an alternative representation for the upper incomplete Gamma
function \cite[Eq. (8.352/2)]{Gradshteyn} and applying the
multinomial theorem, we
obtain%
\begin{equation}\label{gam}
\left[ \Gamma \left( \frac{k}{2},\frac{\mathcal{R}^{2}\left(
\mathbf{H}\right) }{ 2}\rho \right) \right] ^{q}=\exp \left(
-\frac{q\mathcal{R}^{2}\left( \mathbf{H} \right) }{2}\rho \right)
\left( \Gamma \left( \frac{k}{2}\right) \right) ^{q}\Gamma \left(
q+1\right)
\times\sum\limits_{\substack{ n_{0},...,n_{\frac{k}{2}-1}=0 \\ n_{0}+...+n_{%
\frac{k}{2}-1}=q}}^{q}\frac{\left( \frac{\mathcal{R}^{2}\left(
\mathbf{H}\right) }{2}\rho \right) ^{\mathcal{Z}}}{\Psi},
\end{equation}
where $\mathcal{Z}=\sum\limits_{i=0}^{\frac{k}{2}-1}in_{i}$ and
$\Psi=\prod\limits_{i=0}^{\frac{k}{2}-1}\left( \left( i!\right)
^{n_{i}}\Gamma \left( n_{i}+1\right) \right)$.

Hence, (\ref{bound_u}) can be rewritten as
\begin{equation}
A(\rho,N;k,L)  =\sum\limits_{q=0}^{L}\sum_{\substack{ %
n_{0},...,n_{\frac{k}{2}-1}=0 \\ n_{0}+...+n_{\frac{k}{2}-1}=q}}^{q}\frac{L!%
}{\left( L-q\right) !}\frac{\left( -1\right) ^{q}\rho
^{\mathcal{Z}}}{\Psi }\, \mathcal{E}_1, \label{bound_u_2}
\end{equation}%
where
\begin{equation}
\mathcal{E}_1=\mathbb{E}\left[ \left( \frac{\mathcal{R}^{2}\left( \mathbf{H}%
\right) }{2}\right) ^{\mathcal{Z}}\exp \left( -%
\frac{q\mathcal{R}^{2}\left( \mathbf{H}\right) }{2}\rho \right)
\right] . \label{integral_sub}
\end{equation}%
Taking into consideration that the cdf of $\gamma _{\min
}=\min\limits_{i=1,\ldots,N}\gamma_i$ is given by
\begin{equation}
P_{\gamma _{\min }}\left( x\right) =1-\left( \frac{\Gamma \left(
m,mx\right) }{\Gamma \left( m\right) }\right) ^{N},
\label{pdf_min}
\end{equation}%
and after employing \cite[Eq. (8.356.4)]{Gradshteyn}, the pdf of $%
\mathcal{R}^{2}\left( \mathbf{H}\right) $ can be straightforwardly
derived,
according to (\ref{radius_sub}), as%
\begin{equation}
p_{\mathcal{R}^{2}}\left( x\right) =\frac{\left[ \Gamma \left( m,\frac{4m}{%
d_{\min }^{2}}x\right) \right] ^{N-1}x^{m-1}\exp \left(
-\frac{4m}{d_{\min }^{2}}x\right) }{\left( \frac{d_{\min
}^{2}}{4}\right) ^{m}m^{-m}\left[ \Gamma \left( m\right) \right]
^{N}N^{-1}},  \label{pdf_r_sub}.
\end{equation}%
Equivalently, using (\ref{gam}), (\ref{pdf_r_sub}) can be written as%
\begin{equation}
p_{\mathcal{R}^{2}}\left( x\right)=\frac{N\exp \left(
-\frac{4mN}{d_{\min }^{2}}x\right) }{\left( \frac{d_{\min
}^{2}}{4m}\right) ^{-1}\frac{\Gamma
\left( m\right) }{\Gamma \left( N\right) }}\sum_{\substack{ t_{0},...,t_{m-1}=0 \\ t_{0}+...+t_{m-1}=N-1}}%
^{N-1}\frac{\left( \frac{4m}{d_{\min }^{2}}x\right)
^{\mathcal{Y}+m-1}}{\Xi}. \label{pdf_r_sub_2}
\end{equation}
where $\mathcal{Y}=\sum\limits_{j=0}^{m-1}jt_{j}$ and
$\Xi=\prod\limits_{j=0}^{m-1}\left( \left( j!\right)
^{t_{j}}\Gamma \left( t_{j}+1\right) \right)$.

Hence, for the expectation in (\ref{integral_sub}), denoted as
$\mathcal{E}_1$, when $q=0$, it holds that
$\mathcal{Z}=\sum\limits_{i=0}^{\frac{k}{2}-1}in_{i}=0$ and thus
$\mathcal{E}_1=1$, while for $q>0$, $\mathcal{E}$ can be can be
analytically evaluated as \cite[Eq. (3.326.2)]{Gradshteyn}
\begin{equation}\label{integral_2}
\mathcal{E}_1=\frac{Nm^{m}}{\Gamma \left( m\right) }
\sum_{\substack{ t_{0},...,t_{m-1}=0 \\ t_{0}+...+t_{m-1}=N-1}}^{N-1}%
\frac{m^{\mathcal{Y}}\left( \frac{d_{\min }^{2}}{8}\right)
^{\mathcal{Z}}\Gamma
\left(\mathcal{Y}+m+\mathcal{Z}\right)}{\Xi\left( mN+\frac{q\rho
d_{\min }^{2}}{8}\right) ^{\mathcal{Y}+m+\mathcal{Z}}}
\end{equation}%
Finally, by combining (\ref{bound_u_2}) with (\ref{integral_2})
and taking into account the case for $q=0$, it yields
(\ref{cl_f_A}) and this concludes the proof.

\section{Closed Form for the Function $B(\rho,N;k)$}\label{appendixe}
The function
\begin{equation}\label{f_B_2}
\begin{array}{cc} B(\rho,N;k)=\mathbb{E}\left[  1-\frac{\Gamma
\left(
\frac{k}{2},\frac{R_k^{2}\left( \mathbf{H}\right) }{2}\rho \right) }{%
\Gamma \left( \frac{k}{2}\right) }\right] ,&0<k<N,
\end{array}
\end{equation}
can be written, using (\ref{radius}), as
\begin{equation}
B(\rho,N;k)=1-\frac{1}{\Gamma \left( \frac{k}{2}\right)
}\mathbb{E}\left[
\Gamma \left( \frac{k}{2},\frac{\rho \Gamma \left( \frac{k}{2}+1\right) ^{%
\frac{2}{k}}W^{2}}{2\pi }\max_{i=1,...,N}\left( \gamma _{i}\right)
\right) \right],  \label{general1}
\end{equation}%
or if we set $b=\max\limits_{i= 1,...,N}\gamma _{i}$, then
\begin{equation}\label{general2}
B(\rho,N;k)=1-\frac{1}{\Gamma \left( \frac{k}{2}\right) }%
\int_{0}^{\infty }\Gamma \left( \frac{k}{2},\frac{\rho \Gamma \left( \frac{k%
}{2}+1\right) ^{\frac{2}{k}}W^{2}}{2\pi }x\right) f_{b}\left(
x\right) dx.
\end{equation}
The cdf of $b$ is
\begin{equation}
F_{b}\left( x\right) =\left( 1-\frac{\Gamma \left( m,mx\right)
}{\Gamma \left( m\right) }\right) ^{N}
=\sum\limits_{q=0}^{N}\binom{N}{q}\frac{\Gamma \left( m,mx\right) ^{q}}{%
\Gamma \left( m\right) ^{q}}.  \label{cdf_1}
\end{equation}
Now, as in (\ref{gam}), (\ref{cdf_1}) can be rewritten as
\begin{equation}
F_{b}\left( x\right) =\sum\limits_{q=0}^{N}\binom{N}{q}\exp \left(
-qmx\right) \Gamma \left( q+1\right) \sum_{\substack{ %
n_{0},n_{1},...,n_{m-1}=0  \\
n_{0}+n_{1}+...+n_{m-1}=q}}^{q}x^{\mathcal{X}}\Upsilon,
\label{cdf_2}
\end{equation}
where $\mathcal{X}=\sum\limits_{i=0}^{m-1}in_{i}$ and
$\Upsilon=\prod\limits_{i=0}^{m-1}\frac{\left(\frac{m^{i}}{i!}\right)
^{n_{i}}}{\Gamma \left( n_{i}+1\right)}$.

The pdf of $b$ is obtained by taking the derivative of
(\ref{cdf_2}) as
\begin{equation}\label{pdf}
f_{b}\left( x\right) =\sum\limits_{q=0}^{N}\binom{N}{q}\exp \left(
-qmx\right) \Gamma \left( q+1\right)
\sum\limits_{\substack{ %
n_{0},n_{1},...,n_{m-1}=0  \\
n_{0}+n_{1}+...+n_{m-1}=q}}^{q}x^{\left[\mathcal{X}-1\right]}\left(
-qmx+\mathcal{X}\right) \Upsilon.
\end{equation}

Hence, (\ref{general2}) becomes equivalent with%
\begin{equation}\label{general_3}
\begin{array}{ll}
B(\rho,N;k)=1-\sum\limits_{q=0}^{N}\frac{\binom{N}{q}\Gamma \left( q+1\right)}{\Gamma \left( \frac{k}{2}\right) }&\times \left\{\sum\limits_{\substack{ n_{0},n_{1},...,n_{m-1}=1 \\ %
n_{0}+n_{1}+...+n_{m-1}=q \\ \sum\limits_{i=0}^{m-1}in_{i}\neq
0}}^{q}\mathcal{X} \Upsilon f\left( \mathcal{X},\frac{\rho \Gamma
\left( \frac{k}{2}+1\right) ^{\frac{2}{k}}W^{2}}{2\pi },qm,\frac{k}{2}%
\right)\right.
\\
&\left. - \sum\limits_{\substack{ n_{0},n_{1},...,n_{m-1}=1 \\ %
n_{0}+n_{1}+...+n_{m-1}=q}}^{q}qm\Upsilon f\left(
\mathcal{X}+1,\frac{\rho \Gamma \left( \frac{k}{2}%
+1\right) ^{\frac{2}{k}}W^{2}}{2\pi
},qm,\frac{k}{2}\right)\right\},
\end{array}
\end{equation}%
where
\begin{equation}
f\left( \alpha ,\beta ,p,\nu \right) =\int_{0}^{\infty }x^{\alpha
-1}\Gamma \left( \nu ,\beta x\right) \exp \left( -px\right) dx.
\end{equation}%
Using \cite[Eq. (2.10.3.2)]{Prudnikov}, (\ref{general_3}) can be
reduced to (\ref{cl_f_b}) with $g\left( \alpha ,\beta ,p,\nu
\right)$ as defined in (\ref{func_g}), and this concludes the
proof.

\section{Closed Form for the Function $C(\rho,N;L)$}\label{appendixf}
The function
\begin{equation}\label{f_C_2}
C(\rho,N;L)=\mathbb{E}\left[ \left( 1-\frac{\Gamma \left(
\frac{N}{2},\frac{R_N^{2}\left( \mathbf{H}\right) }{2}\rho \right) }{%
\Gamma \left( \frac{N}{2}\right) }\right)^L\right] ,
\end{equation}
can be written, following a similar analysis as in Appendix
\ref{appendixd}, as
\begin{equation}
C(\rho,N;L)  =\sum\limits_{q=0}^{L}\sum_{\substack{ %
n_{0},...,n_{\frac{N}{2}-1}=0  \\ n_{0}+...+n_{\frac{N}{2}-1}=q}}^{q}\frac{%
\frac{L!}{\left( L-q\right) !}\left( -1\right) ^{q}\left(
\frac{1}{2}\rho \right) ^{\mathcal{Q}}}{\Phi }\,\mathcal{E}_2,
\label{bound6}
\end{equation}
where $\mathcal{Q}=\sum\limits_{i=0}^{\frac{N}{2}-1}in_{i}$,
$\Phi=\prod\limits_{i=0}^{\frac{N}{2}-1}\left( \left( i!\right)
^{n_{i}}\Gamma \left( n_{i}+1\right) \right)$ and
\begin{equation}\label{e2}
\mathcal{E}_2\mathcal{=}\mathbb{E}\left[ \left( R_N^{2}\left(
\mathbf{H}\right) \right) ^{\mathcal{Q}}\exp \left(
-\frac{qR_N^{2}\left( \mathbf{H}\right) }{2}\rho \right) \right] .
\end{equation}
Using \cite[Eq. (5)]{Karag_multihop} and after a variable
transformation, the pdf of $R_N^{2}\left( \mathbf{H}\right) $ can
be straightforwardly
obtained for Nakagami fading model as%
\begin{equation}\label{pdf_of_prod}
p_{R_N^{2}}\left( x\right) =\frac{Nx^{-1}}{\left( \Gamma \left(
m\right) \right) ^{N}}
G_{0,N}^{N,0}\left[ \left( \frac{\pi xm}{\left( \Gamma \left( \frac{%
N}{2}+1\right) \right) ^{\frac{2}{N}}}\right) ^{N}\left\vert
\begin{array}{c}
- \\
m,...,m \\
\end{array}%
\right. \right].
\end{equation}
Hence, for the expectation in (\ref{e2}), denoted as
$\mathcal{E}_2$, when $q=0$ it is $\mathcal{E}_2=1$, whereas for
$q\neq 0$, it can be analytically evaluated, by expressing its
integrand $\exp \left( \cdot \right) $ in terms of Meijer's
G-functions according to \cite[Eq. (8.4.3.1)]{Prud} and using
\cite[Eq. (2.24.1.1)]{Prud}, as
\begin{equation}\label{aver2}
\mathcal{E}_2=\frac{\sqrt{N}\left( \frac{2N}{q\rho }\right)
^{\mathcal{Q}}}{\left( \Gamma \left( m\right) \right) ^{N}\left(
2\pi \right) ^{\frac{N-1}{2}}}
G_{N,N}^{N,N}\left[ \frac{\left( \frac{2\pi mN}{q\rho }\right) ^{N}}{%
\left( \Gamma \left( \frac{N}{2}+1\right) \right) ^{2}}\left\vert
\begin{array}{c}
\frac{1-\mathcal{Q}}{N},...,\frac{%
N-\mathcal{Q}}{N} \\
m,...,m \\
\end{array}%
\right. \right].
\end{equation}%
By combining (\ref{aver2}) with (\ref{bound6}), and taking into
account the special case for $q=0$, (\ref{bound6}) can be written
as in (\ref{slb3}) and this concludes the proof.

\newpage

\begin{figure}[h!]
\centering\includegraphics[keepaspectratio,width=6in]{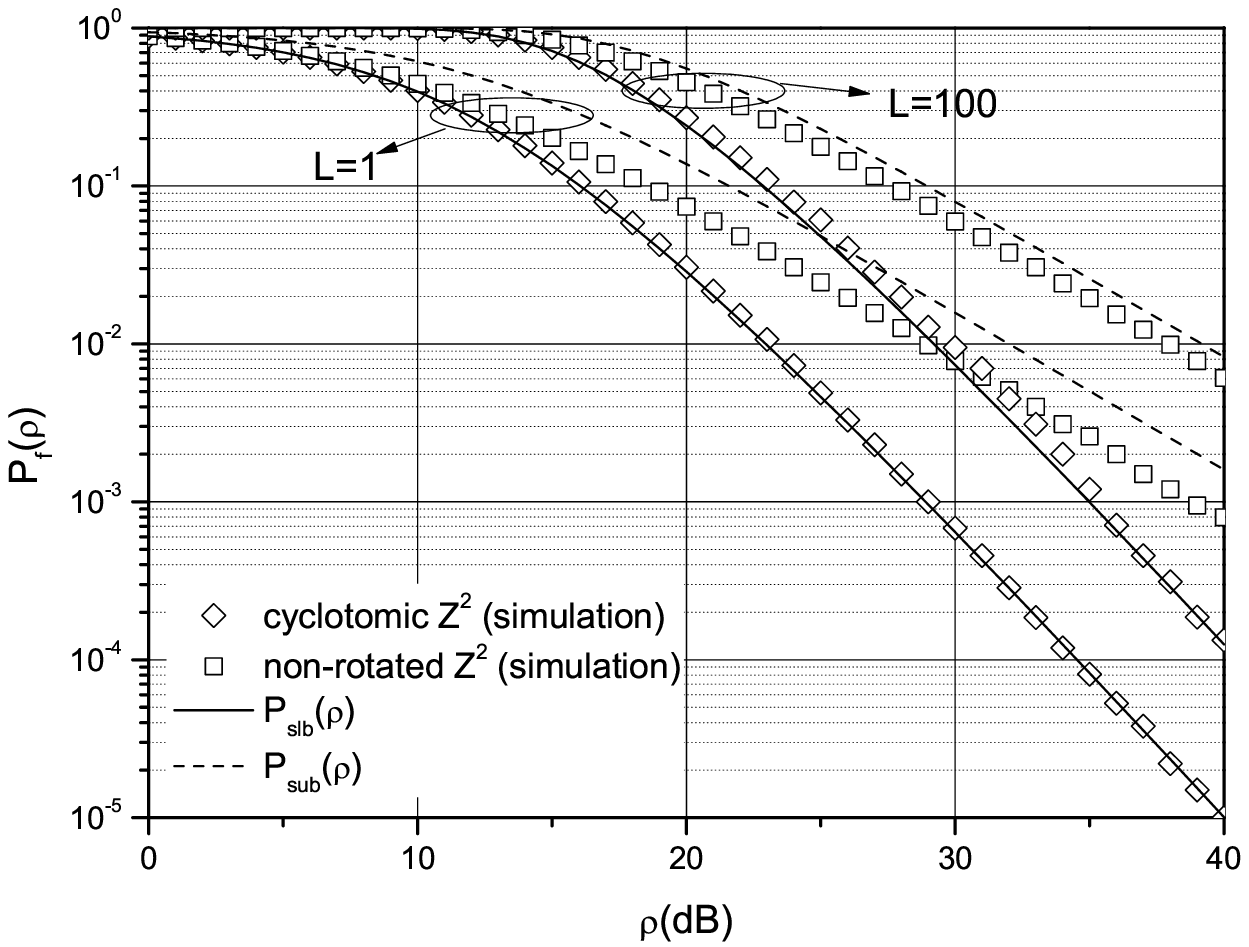}
\center\caption{Frame Error Probability, SLB and SUB for the
$\mathbb{Z}^2$ infinite lattice constellation, for $m=1$ and
$L=1,100$.} \label{Fig:inf_varL}
\end{figure}

\begin{figure}[h!]
\centering\includegraphics[keepaspectratio,width=6in]{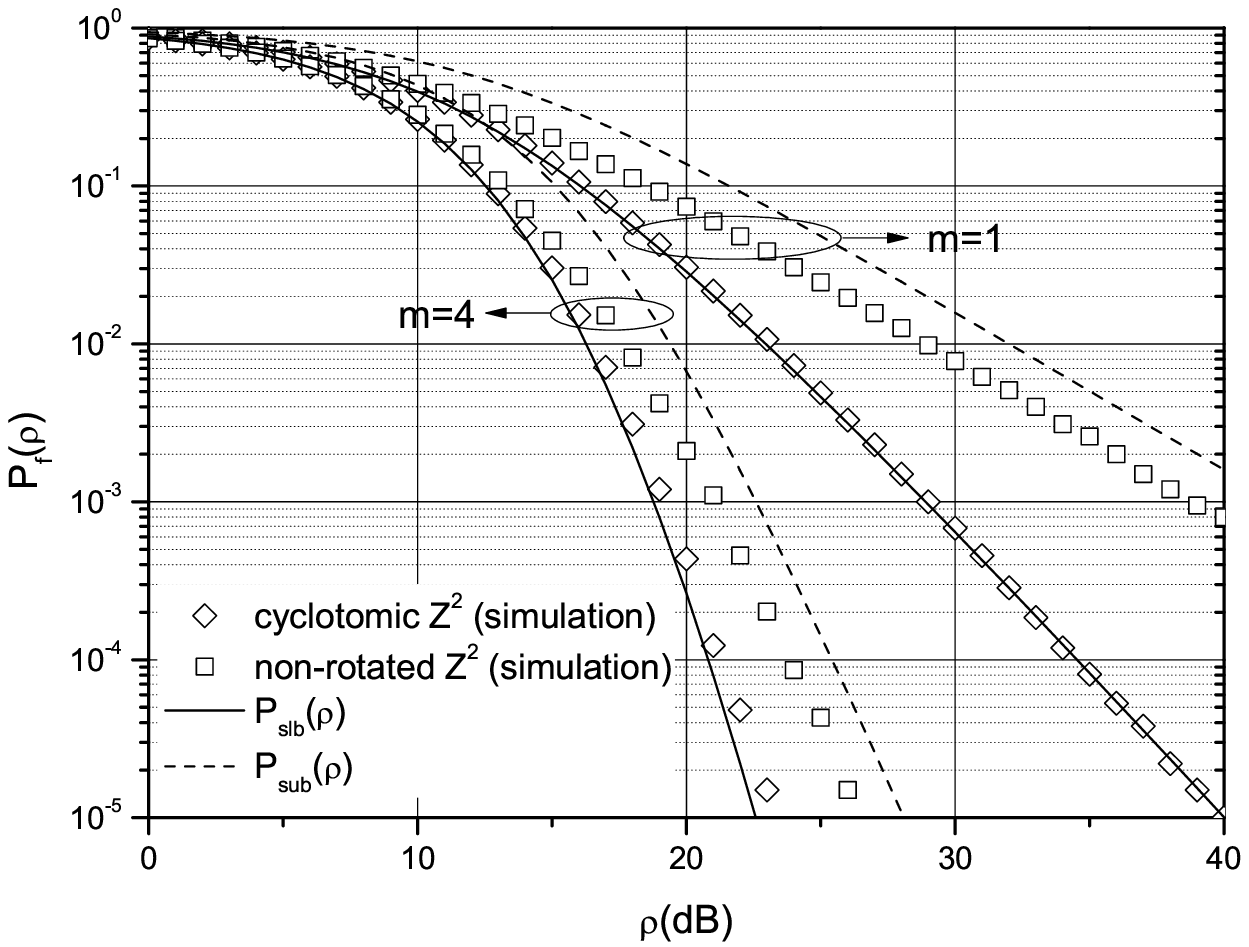}
\center\caption{Frame Error Probability, SLB and SUB for the
$\mathbb{Z}^2$ infinite lattice constellation, for $m=1,4$ and
$L=1$.} \label{Fig:inf_varm}
\end{figure}

\begin{figure}[h!]
\centering\includegraphics[keepaspectratio,width=6in]{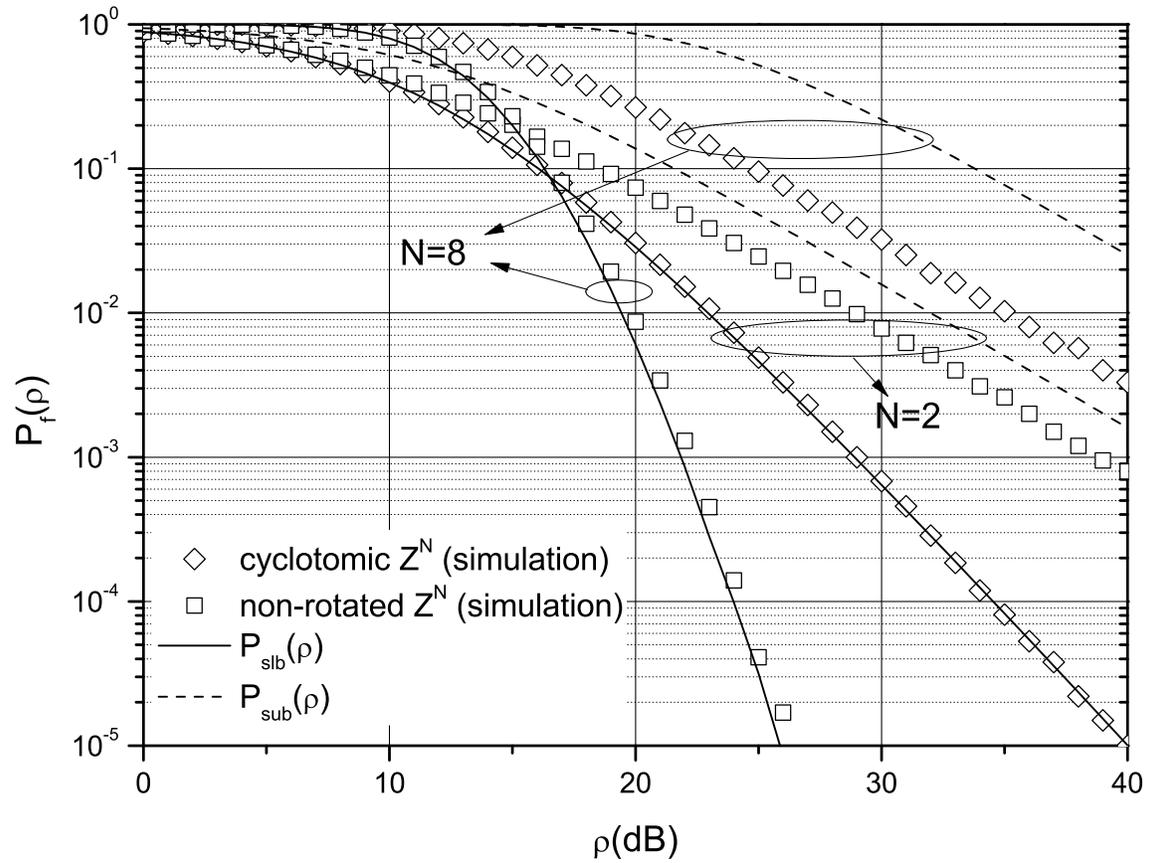}
\center\caption{Frame Error Probability, SLB and SUB for the
$\mathbb{Z}^N$ infinite lattice constellation, for $m=1$ $L=1$ and
$N=2,8$.} \label{Fig:inf_varN}
\end{figure}

\begin{figure}[h!]
\centering\includegraphics[keepaspectratio,width=6in]{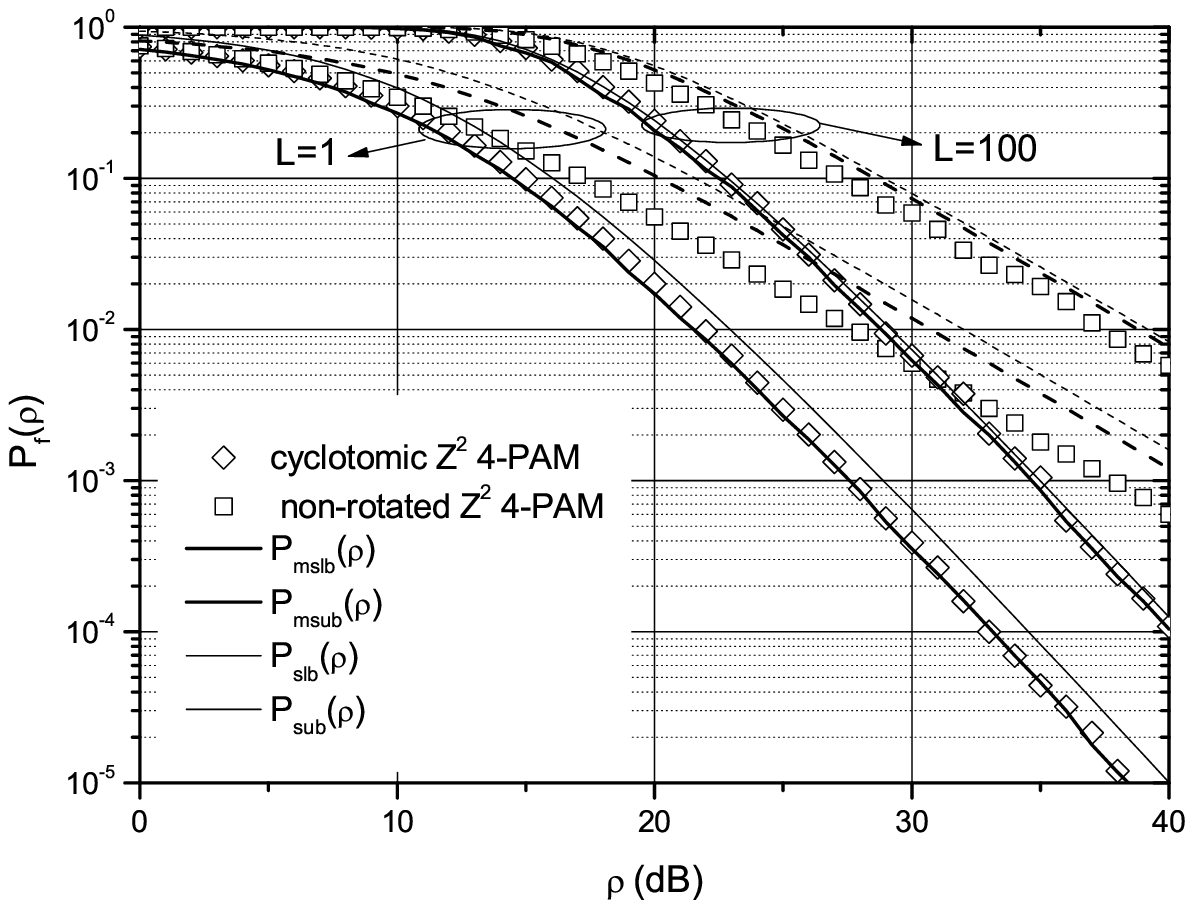}
\center\caption{Frame Error Probability, MSLB, MSUB, SLB and SUB
for the $\mathbb{Z}^2$ $4$-PAM constellation, for $m=1$ and
$L=1,100$.} \label{Fig:fin_varL}
\end{figure}

\begin{figure}[h!]
\centering\includegraphics[keepaspectratio,width=6in]{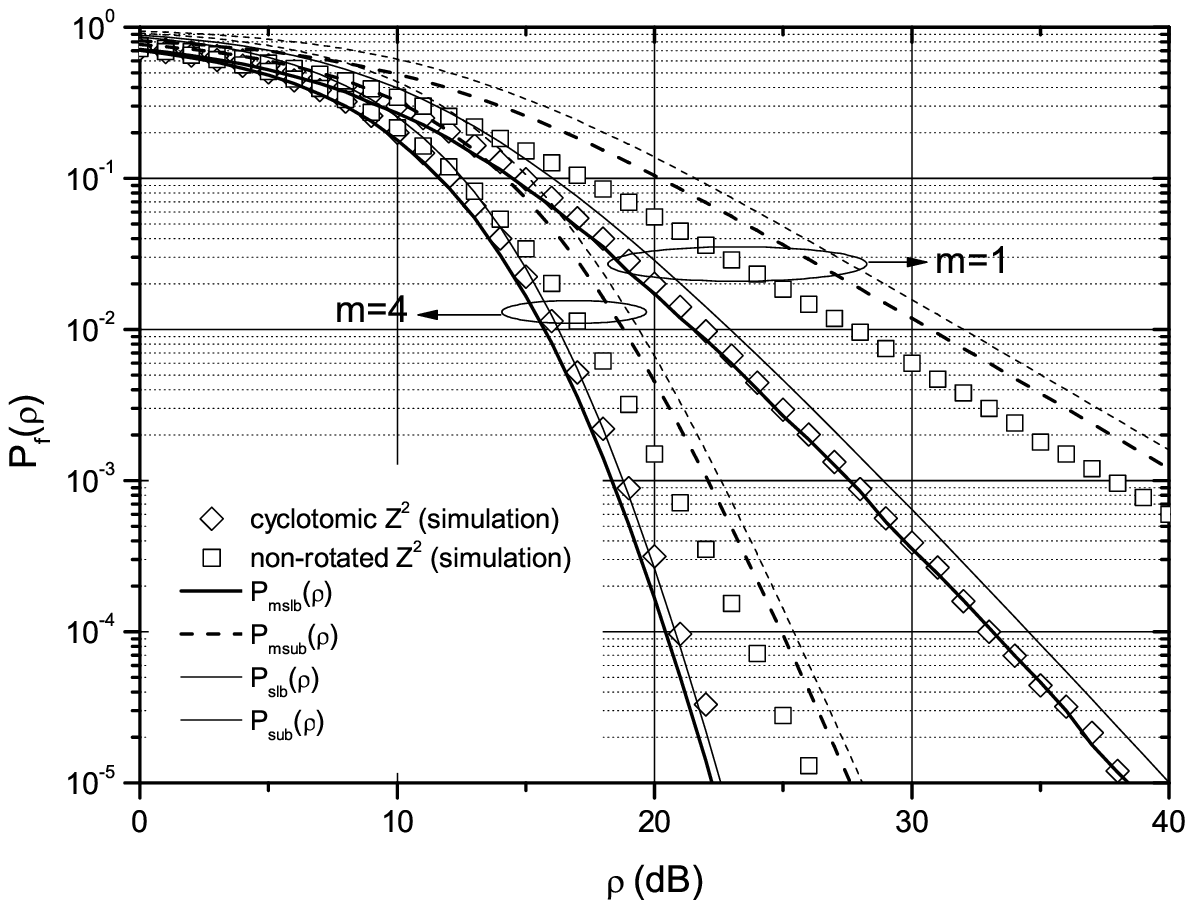}
\center\caption{Frame Error Probability, MSLB, MSUB, SLB and SUB
for the $\mathbb{Z}^2$ $4$-PAM constellation, for $m=1,4$ and
$L=1$.} \label{Fig:fin_varm}
\end{figure}

\begin{figure}[h!]
\centering\includegraphics[keepaspectratio,width=6in]{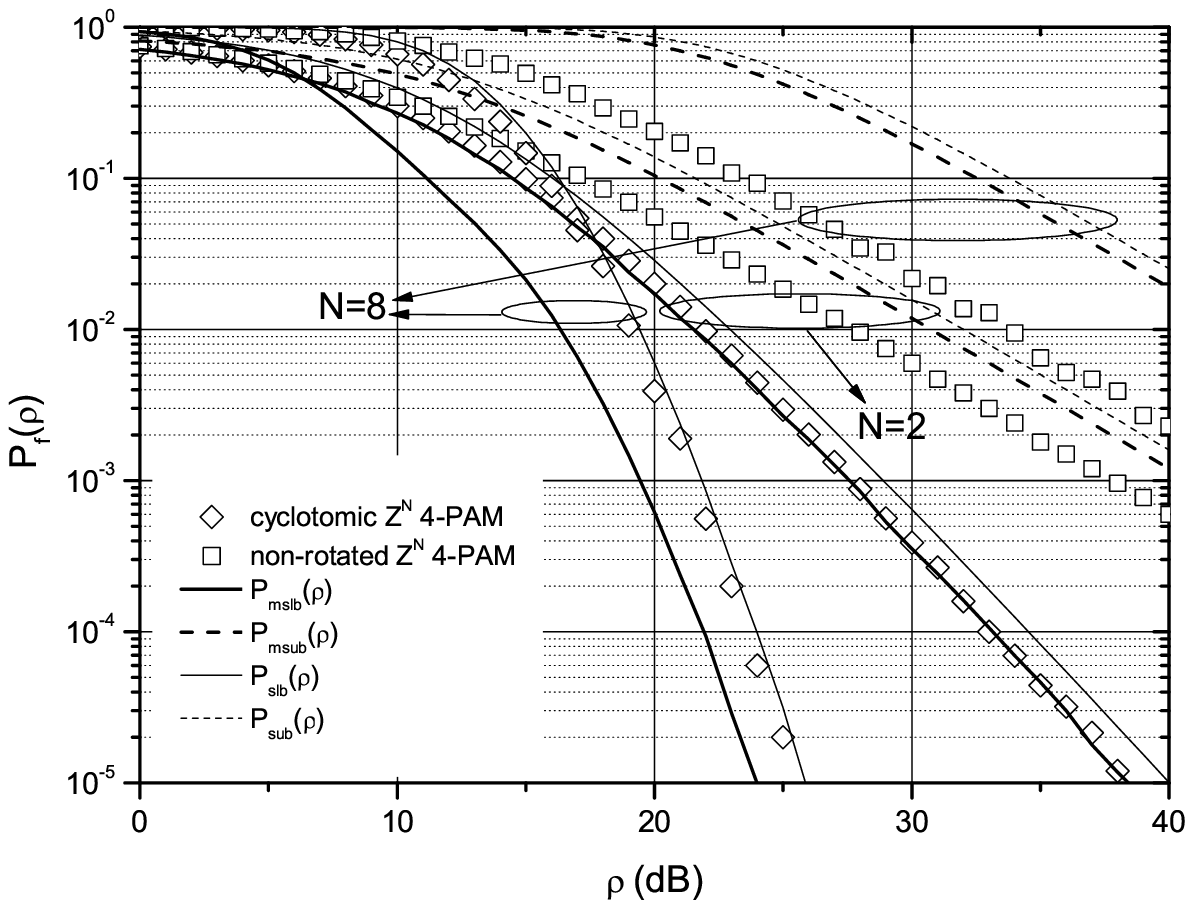}
\center\caption{Frame Error Probability, MSLB, MSUB, SLB and SUB
for the $\mathbb{Z}^N$ $4$-PAM constellation, for $m=1$, $L=1$ and
$N=2,8$.} \label{Fig:fin_varN}
\end{figure}

\begin{figure}[h!]
\centering\includegraphics[keepaspectratio,width=6in]{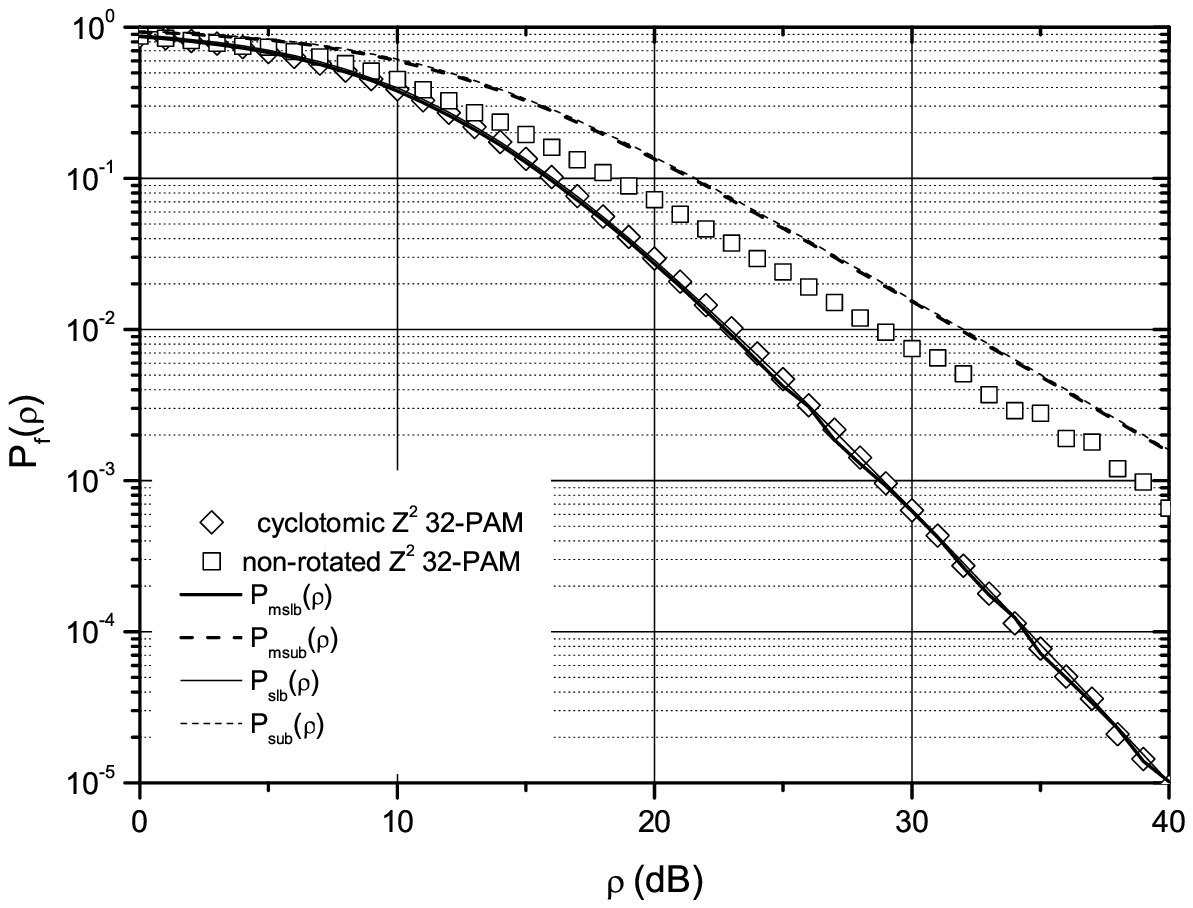}
\center\caption{Frame Error Probability, MSLB, MSUB, SLB and SUB
for the $\mathbb{Z}^2$ $32$-PAM constellation, for $m=1$ and
$L=1$.} \label{Fig:fin_K32}
\end{figure}

\begin{figure}[h!]
\centering\includegraphics[keepaspectratio,width=6in]{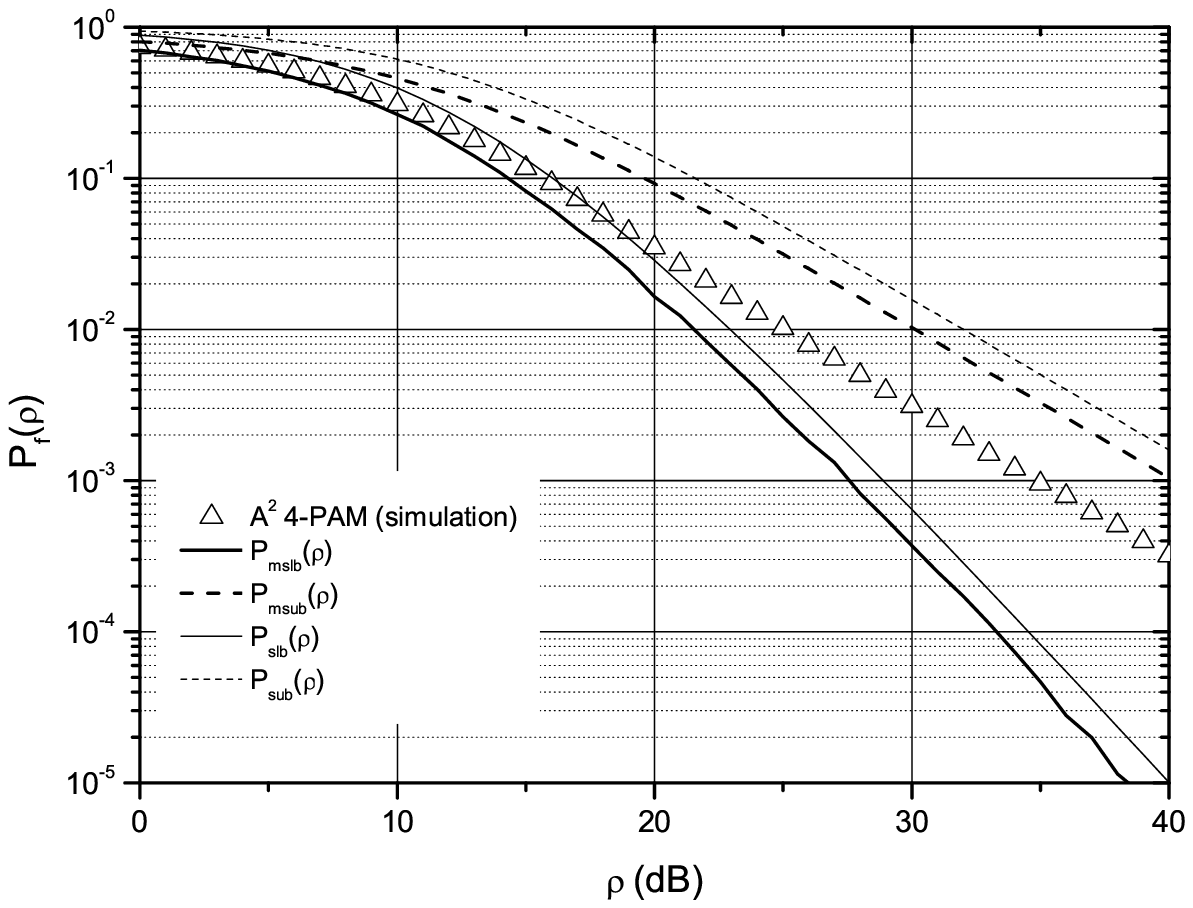}
\center\caption{Frame Error Probability, MSLB, MSUB, SLB and SUB
for the $\mathbb{A}^2$ $4$-PAM constellation, for $m=1$ and
$L=1$.} \label{Fig:fin_A2}
\end{figure}

\end{document}